\documentclass[aps,pra,twocolumn,showpacs,superscriptaddress,groupedaddress]{revtex4}  
\usepackage{graphicx}  
\usepackage{dcolumn}   
\usepackage{bm}        
\usepackage{amsthm}
\usepackage{amsmath}
\usepackage{amssymb}
\usepackage{amsfonts}
\usepackage{amscd}
\usepackage{ascmac}
\usepackage{enumerate}
\usepackage{bm}
\usepackage{latexsym}
\usepackage{color}

\newtheorem{thm} {Theorem}
\newtheorem{lem}[thm] {Lemma}


\def\;{{\hspace{0.3ex};\hspace{0.5ex}}}
\def\,{{\hspace{0,3ex},\hspace{0.5ex}}}
\def\({{\hspace{1.2ex}(}}

\def\QED{\mbox{\rule[0pt]{1.5ex}{1.5ex}}}

\def\endproof{\hspace*{\fill}~\QED\par\endtrivlist\unskip}
 \newenvironment{proofof}[1]{\vspace*{5mm} \par \noindent
{\bf\em Proof of #1:\hspace{2mm}}}{\endproof
}

\def\Label#1{\label{#1} \hbox{[ #1 ]} }
\def\Label{\label}

\begin{document}


\title[Optimal decoy intensity]{Optimal decoy intensity for decoy quantum key distribution}

\author{Masahito Hayashi}

\address{Graduate School of Mathematics, Nagoya University, Furocho, Chikusa-ku, Nagoya, 464-860 Japan}
\address{Centre for Quantum Technologies, National University of Singapore, 3 Science Drive 2, Singapore 117542}

\begin{abstract}
In the decoy quantum key distribution,
we show that a smaller decoy intensity gives a better key generation rate in the asymptotic setting
when we employ only one decoy intensity and the vacuum pulse.
In particular, the counting rate of single photon can be perfectly estimated
when the decoy intensity is infinitesimal. 
The same property holds even when the intensities cannot be perfectly identified. 
Further, we propose a protocol to improve 
the key generation rate over the existing protocol under the same decoy intensity.
\end{abstract}




\maketitle


\section{Introduction}
Quantum key distribution (QKD) by BB84 protocol \cite{BB84} is one of the most important applications of quantum information.
The original QKD requires the single photon source.
However, many economically realizable photon sources produce only weak coherent pulses.
So, they cannot be used for the original QKD protocol.
To solve this problem, we need to
estimate the detection rate of the single photon pulse. 
Hwang \cite{decoy1} proposed the decoy method, in which,
we estimate the detection rate of the single photon pulse from the detection rates of the weak coherent pulses with different intensities.
As another solution, continuous variable quantum key distribution works with weak coherent pulses. (see \cite{2} and references therein)
While continuous variable quantum key distribution can be implemented with an inexpensive Homodyne detection,
the decoy method with BB84 protocol can achieve the longest distance 
with the current technology\cite{Sasa,Stucki}.
So, it is natural to focus on the decoy method.

In the decoy method, we employ two kinds of weak coherent pulses; One is the signal pulse, which generates raw keys.
The other is the decoy pulse, which is used only for the estimation of the detection rate of the single photon pulse.
The key point is the difference between the signal and decoy intensities $\mu_s$ and $\mu_d$, which are the intensities of the signal and decoy pulses.
Using the detection rates of these pulses, the decoy method determines a lower bound of the detection rate of single photon pulse.
However, it cannot uniquely determine this detection rate
although it has been improved by many researchers \cite{decoy2,decoy3,Ma05,Wang05,H1,decoy4,Lo2,wang2,wang3}.
To improve 
this estimation, 
the papers \cite{H1,decoy4} proposed to increase the number of the decoy intensities,
and showed that this detection rate can be uniquely determined when the number of the decoy intensities is infinitely large.
However, it strains the network system of QKD to increase the number of decoy intensities. 
So, it is better to realize a precise estimation without increase of this number. 

In this paper, we focus on the case when we employ only one decoy intensity $\mu_d$ and the vacuum pulse
for the estimation of the detection rate of single photon pulse.
Firstly, we consider a formula for secure key generation rate that 
is different from the conventional one.
Indeed, while the paper \cite{HN} discussed the key generation rate of finite-length setting,
our formula can be regard as the asymptotic version of the key generation rate given in \cite{HN}.
We show that our our formula is better than the conventional one.
Secondly, 
we optimize the choice of the decoy intensity.
This kind of optimization 
for the conventional formula for the asymptotic key generation rate
has been done by Ma et al \cite{Ma05}
when the source intensities are perfectly controlled.
We derive the same optimization for our improved formula for 
the asymptotic key generation rate.
Further, similar to Wang \cite{wang2,wang3}, 
we extend these result to the case when the intensities are different from our intent. 
Even in this generalization, we still have the same conclusion.
On the other hand, Ma et al \cite{Ma05} also considered a similar optimization 
for the conventional formula for the asymptotic key generation rate
when the source intensities have statistical fluctuation.
However, since their setting is different from our setting as explained 
in the second paragraph of Section \ref{s3},
our analysis is different for their analysis.

The remaining part of this paper is organized as follows.
Section \ref{s2} discusses the decoy method when the source intensities are controlled.
Then, we explain our improved formula for the asymptotic key generation rate.
We derive the optimal decoy intensity of this case.
Section \ref{s3} extends the above result to the case when the source intensities cannot be perfectly identified.
Section \ref{s4} discusses the relation of the obtained result with the finite-length case \cite{HN}.
Several proofs are given in Appendixes.

\section{Controlled source intensities}\Label{s2}
First, we discuss the case when the source intensities are controlled.
To discuss this case,
we recall the improved GLLP formula \cite{GLLP,Lo}.
When we distill the secure key from given $M$-bits raw key in the bit basis from the signal pulse,
we firstly apply error correction and then obtain $(1-\eta h(e_{s,+})) M $-bits corrected key,
where $e_{s,+}$ is the error rate in the bit basis of the signal pulse.
Here, $h$ is the binary entropy with the logarithm to the base 2
and the parameter $\eta$ is the efficiency of error correction, which is chosen to be $1.1$ in a realistic case and to be $1$ in the ideal case.
The next step, the privacy amplification, requires
the ratio $p:q:r$ ($p+q+r=1$) of the vacuum pulse, the single photon pulse, and the multi-photon pulse
among the received pulses.
When the error rate of the phase basis in the single photon pulse is $e_{\times}$,
it is enough to apply universal2 hash functions sacrificing $(q h(e_{\times})+ r)M$ bits \cite{Renner, H2, finite}.
Hence, we can obtain $(1-\eta h(e_{s,+})-q h(e_{\times})- r ) M  $ bits of secure key.
That is, the secure key generation rate 
per received pulse with the matched basis is 
$1-\eta h(e_{s,+})-q h(e_{\times})- r$.
The rates $q$ and $r$ can be calculated 
from the detection rates $a$, $p_0$, and $p_{s,+} $ of the single photon pulse, the vacuum pulse 
and the signal pulse 
as follows.
Since the signal pulse has the intensity $\mu_s$,
the transmitted signal pulses consist of 
the vacuum pulse, the single photon pulse, and the multi-photon pulse
with the ratio 
$e^{-\mu_s}:\mu_s e^{-\mu_s}:1- (1+\mu_s)e^{-\mu_s}$.
Then, the ratio of the vacuum pulse, the single photon pulse, and the multi-photon pulse among detected signal pulses 
is $p_0 e^{-\mu_s}:a \mu_s e^{-\mu_s}:p_{s,+} - (p_0+a \mu_s)e^{-\mu_s}$.
That is, $q=a \mu_s e^{-\mu_s}/p_{s,+}$ and $r=1 - (p_0+a \mu_s)e^{-\mu_s}/p_{s,+}$.
Therefore, the secure key generation rate 
per received pulse with the matched basis is 
\begin{align}
\frac{(p_0+a \mu_s(1-h(e_{\times})) )e^{-\mu_s}}{p_{s,+}}-\eta h(e_{s,+}).
\Label{a2}
\end{align}
However, the rate $a$ and the phase error rate $e_{\times}$ cannot be directly measured
although $p_0$ can be directly measured by transmitting the vacuum pulse.

The decoy method enables us to estimate the quantities $a$ and $e_{\times}$ by using the above measurable values
and the detection rates $p_{s,\times}$ and $p_{d,\times}$ of the signal and decoy pulses with the phase basis,
and the error rates $e_{s,\times}$ and $e_{d,\times}$ of the signal and decoy pulses with the phase basis.
These rates can be measured by randomizing the basis and the intensity.

In the existing method \cite{decoy3,Ma05,Wang05,H1,decoy4},
they derive the estimate $\hat{a}$ of the detection rate $a$ of single-photon pulse
and the estimate $\hat{b}$ of the rate $b$ that
the single-photon pulse is detected with the phase error 
as follows.
\begin{align}
&\hat{a}(p_{d,\times},p_{s,\times})\nonumber \\
:=&
\left[\frac{
\mu_s^2 e^{\mu_d}(p_{d,\times} - p_0 e^{-\mu_d}/2)
-
\mu_d^2 e^{\mu_s}(p_{s,\times} - p_0 e^{-\mu_d}/2)
}{\mu_d \mu_s (\mu_s-\mu_d) } \right]_+
\Label{a1}
\\
&\hat{b}_{\times}(e_{d,\times}, p_{d,\times}) 
:= \left[\frac{e_{d,\times}p_{d,\times} e^{\mu_d} -p_0/2 }{\mu_d}\right]_+,
\Label{a3}
\end{align}
where we assume that $\mu_d <\mu_s$.
The above estimates of the case $\mu_d >\mu_s$
can be derived by exchanging the roles of the decoy and signal pulses
in the right hand side.
The key point of the derivation of (\ref{a1}) and (\ref{a3}) 
is the following expansions of the states of the decoy and signal pulses:
$\sum_{n=0}^\infty e^{-\mu_d}\frac{\mu_d^n}{n!}
|n \rangle\langle n|=
e^{-\mu_d}|0\rangle\langle 0|
+e^{-\mu_d}\mu_d|1\rangle\langle 1| 
+(1-(1+\mu_d)e^{-\mu_d} ) \rho_2$, and
$\sum_{n=0}^\infty e^{-\mu_s}\frac{\mu_s^n}{n!}
|n \rangle\langle n|=
e^{-\mu_s}|0\rangle\langle 0|
+e^{-\mu_s}\mu_s|1\rangle\langle 1| 
+\beta_2 \rho_2 +\beta_3 \rho_3$,
where the state $\rho_3$ is chosen properly.
The estimate $\hat{a}$ in (\ref{a1}) is derived from the non-negativity of the detection rate $a_3$ of the state $\rho_3$,
and the estimate $\hat{b}$ in (\ref{a3}) is from the non-negativity of the rate 
that the pulse with the state $\rho_2$ is detected with the phase error.
Substituting $\hat{a}$ and $\frac{\hat{b}}{\hat{a}}$
into $a$ and $e_{\times}$ of the formula (\ref{a2}),
we obtain the key generation rate.

In this paper, instead of $\hat{a}$,
we estimate the rate $c$ that the single-photon pulse is detected without the phase error.
Since a smaller $c$ gives a better case,
we can estimate $c$ in the same way as the detection rate $a$.
Then, we obtain the estimate $\hat{c}$ as
\begin{align}
&\hat{c}(p_{d,\times},p_{s,\times}, e_{d,\times},e_{s,\times})
\nonumber \\
:=&
\hat{a}((1-e_{d,\times})p_{d,\times},(1-e_{s,\times})p_{s,\times}),
\Label{a4}
\end{align}
which is derived from the non-negativity of the rate 
$c_3$ that the pulse with the state $\rho_3$ 
is detected without the phase error.
Since the rate $c_3$ is smaller than the rate $a_3$,
the non-negativity of $c_3$ is a stronger 
constraint than that of $a_3$.
So, the estimate $\hat{c}$ in (\ref{a4}) is better than the estimate given in (\ref{a1}).
Therefore, substituting $\hat{c}+\hat{b}$ and $\frac{\hat{b}}{\hat{b}+\hat{c}}$
into $a$ and $e_{\times}$ of the formula (\ref{a2}),
we obtain a better key generation rate.

\begin{figure}[h]
\centering
\includegraphics[width=8.00cm, clip]{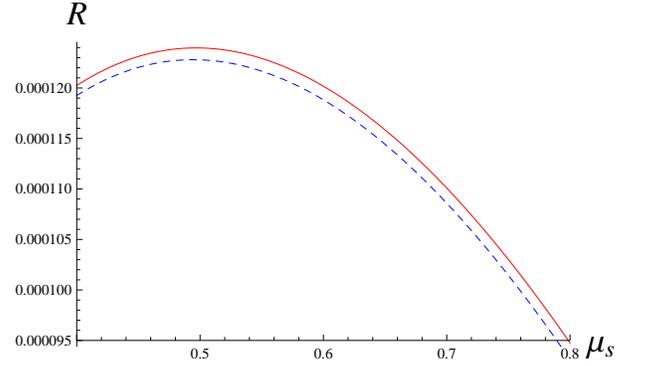}
\caption{Our key generation rate $R(\mu_s,\mu_d)$ (red, normal) and  
existing key generation rate $\tilde{R}(\mu_s,\mu_d)$ \cite{decoy3,Ma05,Wang05,H1,decoy4} (blue, dashed)
with the decoy intensity ${\mu}_d=0.1$.
The horizontal axis expresses the signal intensity ${\mu}_d$.}
\Label{FR4}
\end{figure}

Now, we consider the case with no eavesdropper, i.e.,
the case when the true intensities coincide with our intent intensities.
In the following, 
the subscript $s$ expresses the signal pulse,
and 
the subscript $d$ expresses the decoy pulse.
Then, we adopt the following model 
for the detection rates $p_{i,+},p_{i,\times}$ and the error rates 
$e_{i,+},e_{i,\times}$ with the parameters $\alpha$ and $s$ \cite{LB,GRTZ}:
\begin{align}
&p_{i,+}=p_{i,\times}= 1-e^{-\alpha \mu_i} + p_0 ,\nonumber \\
&e_{i,+}p_{i,+} =e_{i,\times}p_{i,\times}
= s(1-e^{-\alpha \mu_i}) + \frac{p_0}{2} ,\quad
i=s,d,
\Label{10-8-14}
\end{align}
where $\alpha$ is the total transmission including quantum efficiency of the detector,
and $s$ is the error due to the imperfection of the optical system.

Under this assumption, we estimate 
the detection rate $a$ of the single photon pulse, 
the rate $b$ that the single-photon pulse is detected with the phase error, 
and the rate $c$ that the single-photon pulse is detected without the phase error.
By letting $\mu_2$ be the larger intensity of $\mu_s$ and $\mu_d$ is 
and $\mu_1$ be the smaller one,
their estimates $\hat{a}$, $\hat{b}$, and $\hat{c}$ are given as
\begin{widetext}
\begin{align}
\hat{a}(\mu_1,\mu_2) 
& :=
\frac
{
\mu_2^2 e^{\mu_1}
((1-e^{-\alpha \mu_1})+ p_0 (1- e^{-\mu_1}))
-
\mu_1^2 e^{\mu_2}
((1-e^{-\alpha \mu_2})+ p_0 (1- e^{-\mu_2}))
}
{\mu_1 \mu_2 (\mu_2-\mu_1) } \nonumber \\ 
&=
(\mu_2 \mu_1)
\frac
{
\frac{(1+p_0) e^{\mu_1} -e^{ (1-\alpha) \mu_1}-p_0}{\mu_1^2}
-
\frac{(1+p_0) e^{\mu_2} -e^{ (1-\alpha) \mu_2}-p_0}{\mu_2^2}
}
{\mu_2-\mu_1 } 
\Label{10-8-7} \\
\hat{b}(\mu_1) 
&:=
\frac{s(1-e^{-\alpha \mu_1})e^{\mu_1} +\frac{p_0}{2}(e^{\mu_1}-1)}{\mu_1}
\Label{10-8-9-c}\\
\hat{c}(\mu_1,\mu_2)
& :=
\frac
{
\mu_2^2 e^{\mu_1}
((1-s)(1-e^{-\alpha \mu_1})+ \frac{p_0}{2} (1- e^{-\mu_1}) )
-
\mu_1^2 e^{\mu_2}
((1-s)(1-e^{-\alpha \mu_2})+ \frac{p_0}{2} (1- e^{-\mu_2}) )
}
{\mu_1 \mu_2 (\mu_2-\mu_1) }
\nonumber
\\
& =
(\mu_2 \mu_1)
\frac
{
\frac{(1-s+p_0/2) e^{\mu_1} -(1-s)e^{ (1-\alpha) \mu_1}-\frac{p_0}{2}}{\mu_1^2}
-
\frac{(1-s+p_0/2) e^{\mu_2} -(1-s)e^{ (1-\alpha) \mu_2}-\frac{p_0}{2}}{\mu_2^2}
}
{\mu_2-\mu_1 } 
\Label{10-8-9}.
\end{align}
By using these estimates,
the key generation rates $R(\mu_s,\mu_d)$
and $\tilde{R}(\mu_s,\mu_d)$ are written as
\begin{align}
 R(\mu_s,\mu_d)= &
\left\{
\begin{array}{ll}
\mu_s e^{-\mu_s} 
(\hat{c}(\mu_d,\mu_s)+\hat{b}(\mu_d))
 (1-h(
\frac{\hat{b}(\mu_d)}{\hat{c}(\mu_d,\mu_s)+\hat{b}(\mu_d)}
)
) 
+ e^{-\mu_s} p_0 
- p_{s,+} \eta h (\frac{s_{s,+}}{p_{s,+}})
& \hbox{if } \mu_s < \mu_d
\\
\mu_s e^{-\mu_s} 
(\hat{c}(\mu_s,\mu_d)+\hat{b}(\mu_s))
 (1-h(
\frac{\hat{b}(\mu_s)}{\hat{c}(\mu_s,\mu_d)+\hat{b}(\mu_s)}
)
) 
+ e^{-\mu_s} p_0 
- p_{s,+} \eta h (\frac{s_{s,+}}{p_{s,+}})
& \hbox{if } \mu_s > \mu_d.
\end{array}
\right.\Label{10-23-3} \\
\tilde{R}(\mu_s,\mu_d)= &
\left\{
\begin{array}{ll}
\mu_s e^{-\mu_s} 
\hat{a}(\mu_d,\mu_s)
 (1-h(
\frac{\hat{b}(\mu_d)}{\hat{a}(\mu_d,\mu_s)}
)
) 
+ e^{-\mu_s} p_0 
- p_{s,+} \eta h (\frac{s_{s,+}}{p_{s,+}})
& \hbox{if } \mu_s < \mu_d
\\
\mu_s e^{-\mu_s} 
\hat{a}(\mu_s,\mu_d)
 (1-h(
\frac{\hat{b}(\mu_s)}{\hat{a}(\mu_s,\mu_d)}
)
) 
+ e^{-\mu_s} p_0 
- p_{s,+} \eta h (\frac{s_{s,+}}{p_{s,+}})
& \hbox{if } \mu_s > \mu_d.
\end{array}
\right.
\Label{10-23-4}
\end{align}
\end{widetext}

Then, we obtain the following lemma, which will be shown in Appendix \ref{as1}.
\begin{lem}\Label{thm2}
$(\hat{c}(\mu_1,\mu_2) 
+\hat{b}(\mu_1) )
(1-h(
\frac{\hat{b}(\mu_1)}
{\hat{c}(\mu_1,\mu_2 ) 
+\hat{b}(\mu_1) }
))$
is monotonically decreasing for $\mu_1$ and $\mu_2$
when
$\frac{\hat{b}(\mu_1)}
{\hat{c}(\mu_1,\mu_2 ) 
+\hat{b}(\mu_1) }
< \frac{1}{2}$.
Similarly,
$\hat{a}(\mu_1,\mu_2) 
(1-h(
\frac{\hat{b}(\mu_1)}
{\hat{a}(\mu_1,\mu_2 ) }
))$
is monotonically decreasing for $\mu_1$ and $\mu_2$
when
$\frac{\hat{b}(\mu_1)}
{\hat{a}(\mu_1,\mu_2 ) }
< \frac{1}{2}$.
\end{lem}

Now, we fixed a signal intensity to be $\mu_s$.
Then, Lemma \ref{thm2} implies
\begin{align*} 
R(\mu_s,\mu_1) \ge 
R(\mu_s,\mu_2) \ge 
R(\mu_s,\mu_3) \ge 
R(\mu_s,\mu_4) 
\end{align*}
for $\mu_1 <\mu_2 <\mu_s <\mu_3 <\mu_4$.
These inequalities imply that
a smaller decoy intensity has a better key generation rate
when the signal intensity is fixed.
The same relation holds for $\tilde{R}(\mu_s,\mu_d)$.
Therefore, we obtain the following theorem.
\begin{thm}\Label{12-26-1}
$\tilde{R}(\mu_s,\mu_d)$ and 
$R(\mu_s,\mu_d)$ are monotonically decreasing with respect to $\mu_d$
for a given $\mu_s$.
\end{thm}

Note that 
although the argument for $\tilde{R}(\mu_s,\mu_d)$ in Theorem \ref{12-26-1} was shown in \cite{Ma05},
that for ${R}(\mu_s,\mu_d)$ was not shown in \cite{Ma05}.
This theorem implies that 
a smaller decoy intensity yields a larger key generation rate.
In particular, 
the optimal decoy intensity is infinitesimal small.
Then, the following lemma holds, which will be shown in Appendix \ref{as1}.
\begin{lem}\Label{thm22}
We also obtain
\begin{align}
\lim_{\mu_1 \to +0}\hat{c}(\mu_1,\mu_2) 
&= (1-s)\alpha + \frac{p_0}{2} , \\
\lim_{\mu_1 \to +0}\hat{b}(\mu_1) 
&= s \alpha + \frac{p_0}{2} ,\\
\lim_{\mu_1 \to +0}\hat{a}(\mu_1,\mu_2) 
&= \alpha + p_0 .
\end{align}
\end{lem}

Substituting the above in (\ref{10-23-3}) and (\ref{10-23-4}), 
we have
\begin{align*}
&\lim_{\mu_d \to 0} R(\mu_s,\mu_d)
=\lim_{\mu_d \to 0} \tilde{R}(\mu_s,\mu_d)\\
=&\mu_s e^{-\mu_s} 
(\alpha + p_0) (1-h(\frac{\alpha s + \frac{p_0}{2}}{\alpha+p_0})) 
+ e^{-\mu_s} p_0 \\
&- (1-e^{-\alpha \mu_s} + p_0) h (\frac{s(1-e^{-\alpha \mu_s}) +\frac{p_0}{2}}{1-e^{-\alpha \mu_s} + p_0 }).
\end{align*}
That is, in the limit $\mu_d \to 0$, 
we can perfectly estimate the parameters $a$, $b$, and $c$.
Since the signal intensity is related to the detection rate of the signal pulse and other factors,
it is not so simple to find the optimal signal intensity.

\section{Uncontrolled source intensities}\Label{s3}
Next, we consider the case when 
we cannot perfectly identify 
the true intensities ${\mu}_s$ and ${\mu}_d$.
Similar to Wang et al.\cite{wang2,wang3},
we assume that
the true intensities 
${\mu}_s$ and ${\mu}_d$ belong to certain intervals 
$[(1-\epsilon)\tilde{\mu}_s,(1+\epsilon)\tilde{\mu}_s]$ and $[(1-\epsilon)\tilde{\mu}_d,(1+\epsilon)\tilde{\mu}_d]$
with the error ratio $\epsilon>0$, respectively.
In this case, we have to consider the worst case 
with respect to the true intensities ${\mu}_s$ and ${\mu}_d$ 
in the intervals 
$[(1-\epsilon)\tilde{\mu}_s,(1+\epsilon)\tilde{\mu}_s]$ and $[(1-\epsilon)\tilde{\mu}_d,(1+\epsilon)\tilde{\mu}_d]$, respectively.
Indeed, 
the smaller intensity pulse is generated by the combination of
the stronger pulse and beam splitter.
If the beam splitter is well installed,
the error only comes from the error of the stronger pulse source.
In this assumption, 
the error ratio $\epsilon$ does not depend on the intensity.

Here, we should remark the relation with the setting in 
Ma et al \cite{Ma05}.
They studied the case with the statistical fluctuation of 
the measurement outcomes \cite[Section IV]{Ma05}.
However, we assume the source intensity
is fixed but is different from our intent, and
infinitely large data is available.  
That is, in our setting, there is no statistical fluctuation in our data.
Hence, our model is simpler than their model.
Although they could not obtain an analytical result in their model \cite[Section IV]{Ma05},
we derive an analytical result in our model as follows.

Now, we treat the typical case when 
true intensities are $\tilde{\mu}_1$ and $\tilde{\mu}_2$
and there is no eavesdropper.
Instead of (\ref{10-8-14}),
we assume that 
\begin{align}
p_{i,+}&=p_{i,\times}= 1-e^{-\alpha \tilde{\mu}_i} + p_0 ,\\
e_{i,+}p_{i,+} &=e_{i,\times}p_{i,\times}
= s(1-e^{-\alpha \tilde{\mu}_i}) + \frac{p_0}{2} ,
\Label{10-8-2}
\end{align}
for $i=s,d$.
When we consider that the true signal and decoy intensities are
$\mu_s$ and $\mu_d$,
the detection rate of the single photon pulse is $a$,
the rate that the single-photon pulse is detected with the phase error 
is $b$, 
and
the rate that the single-photon pulse is detected without the phase error 
is $c$, 
the two key generation rates are given as
\begin{align}
\tilde{R}
&=
{(p_0+{a} \mu_s(1-h(\frac{{b}}{{a}})) )e^{-\mu_s} - p_{s,+}\eta h(e_{s,+})
} \\
R
&={(p_0+({b}+{c}) \mu_s(1-h(\frac{{b}}{{b}+{c}})) )e^{-\mu_s} - p_{s,+}\eta h(e_{s,+})
} .
\end{align}
Then, using the functions
\begin{align*}
&f_a(p_1,p_2,\mu_1,\mu_2) \\
:=&
\left[\frac{
\mu_2^2 e^{\mu_1}(p_1 - p_0 e^{-\mu_1})
-
\mu_1^2 e^{\mu_2}(p_2 - p_0 e^{-\mu_2})
}{\mu_1 \mu_2 (\mu_2-\mu_1) } \right]_+
\\
&f_b(s_1,\mu_1) 
:= \left[\frac{s_{1} e^{\mu_1} -p_0/2 }{\mu_1}\right]_+ \\
&f_c(p_1,p_2,\mu_1,\mu_2) \\
:=&
\left[\frac{
\mu_2^2 e^{\mu_1}(p_1 - p_0 e^{-\mu_1}/2)
-
\mu_1^2 e^{\mu_2}(p_2 - p_0 e^{-\mu_2}/2)
}{\mu_1 \mu_2 (\mu_2-\mu_1) } \right]_+,
\end{align*}
we can estimate the parameters $a$, $b$, and $c$ as
\begin{align}
 \hat{a}&=
\left\{
\begin{array}{ll}
f_a( p_{d,\times},p_{s,\times},\mu_d,\mu_s) & \hbox{if } \mu_d < \mu_s \\
f_a( p_{s,\times},p_{d,\times},\mu_s,\mu_d) & \hbox{if } \mu_d > \mu_s 
\end{array}
\right.
\Label{11-3-3} \\
 \hat{b}&=
\left\{
\begin{array}{ll}
f_b(e_{d,\times} p_{d,\times},\mu_d) & \hbox{if } \mu_d < \mu_s \\
f_b(e_{s,\times} p_{s,\times},\mu_s) & \hbox{if } \mu_d > \mu_s 
\end{array}
\right.
\Label{11-3-2},\\
 \hat{c}&=
\left\{
\begin{array}{ll}
f_c((1-e_{d,\times}) p_{d,\times},(1-e_{s,\times}) p_{s,\times},\mu_d,\mu_s) 
& \hbox{if } \mu_d < \mu_s \\
f_c((1-e_{s,\times}) p_{s,\times},(1-e_{d,\times}) p_{d,\times},\mu_s,\mu_d) 
& \hbox{if } \mu_d > \mu_s 
\end{array}
\right.
\Label{11-3-1} .
\end{align}
Thus, when we consider that the true signal and decoy intensities are
$\mu_s$ and $\mu_d$,
by using the above estimates $\hat{a}$, $\hat{b}$, and $\hat{c}$, 
the two key generation rates are given as
\begin{align*}
&{R}_e(\mu_s,\mu_d,\tilde{\mu}_s,\tilde{\mu}_d) \\
:=& \mu_s e^{-\mu_s} 
(\hat{c}+\hat{b})
 (1-h(\frac{\hat{b}}{\hat{c}+\hat{b}})
) 
+ e^{-\mu_s} p_0 
- p_{s,+} \eta h ( e_{s,+}),\\
& \tilde{R}_e(\mu_s,\mu_d,\tilde{\mu}_s,\tilde{\mu}_d) \\
:= &\mu_s e^{-\mu_s} 
\hat{a} (1-h(\frac{\hat{b}}{\hat{a}})) 
+ e^{-\mu_s} p_0 - p_{s,+} \eta h ( e_{s,+}).
\end{align*}
Therefore, 
by taking the worst case, 
the key generation rates are given by
\begin{align*}
R_e(\tilde{\mu}_s,\tilde{\mu}_d,\epsilon) 
:=&
\min_{\mu_i \in [(1-\epsilon)\tilde{\mu}_i,(1+\epsilon)\tilde{\mu}_i], i=s,d } 
R_e(\tilde{\mu}_s,\tilde{\mu}_d,{\mu}_s,{\mu}_d) \\
\tilde{R}_e(\tilde{\mu}_s,\tilde{\mu}_d,\epsilon) 
:=&
\min_{\mu_i \in [(1-\epsilon)\tilde{\mu}_i,(1+\epsilon)\tilde{\mu}_i], i=s,d } 
R_e(\tilde{\mu}_s,\tilde{\mu}_d,{\mu}_s,{\mu}_d) .
\end{align*}
Indeed, it is quite difficult to find 
the values $\mu_s \in [(1-\epsilon)\tilde{\mu}_s,(1+\epsilon)\tilde{\mu}_s],
\mu_d \in [(1-\epsilon)\tilde{\mu}_d,(1+\epsilon)\tilde{\mu}_d]$
realizing the above minimums. 
However, our numerical demonstration (Fig. \ref{f5})
suggests the following
when $\epsilon >0$ is sufficiently small.
When $(1+\epsilon)\tilde{\mu}_d< (1-\epsilon)\tilde{\mu}_s$,
$\mu_d=(1+\epsilon)\tilde{\mu}_d$ and
$\mu_s=(1-\epsilon)\tilde{\mu}_s$ give the minimums.
When $(1+\epsilon)\tilde{\mu}_s< (1-\epsilon)\tilde{\mu}_d$,
$\mu_s=(1+\epsilon)\tilde{\mu}_s$ and
$\mu_d=(1-\epsilon)\tilde{\mu}_d$ give the minimums.
In the remaining case, we cannot distinguish two intensities.
So, the decoy method does not work.

\begin{figure}[htbp]
\centering
\includegraphics[width=8.00cm, clip]{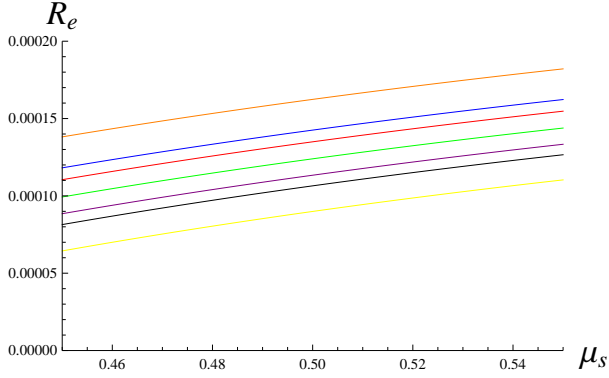}
\caption{Graphs of $R_e(0.5,0.1,{\mu}_s,{\mu}_d)$.
Here, the true decoy intensity $\mu_d$ is chosen to be
$0.90 \times 0.1$ (orange), 
$0.95 \times 0.1$ (blue), 
$0.97 \times 0.1$ (red), 
$1.00 \times 0.1$ (green), 
$1.03 \times 0.1$ (purple), 
$1.05 \times 0.1$ (black), 
and $1.10 \times 0.1$ (yellow).}
\Label{f5}
\end{figure}

Indeed, as will be shown in Theorem \ref{thm20},
$\mu_d=(1+\epsilon)\tilde{\mu}_d$ 
and
$\mu_s=(1-\epsilon)\tilde{\mu}_s$ 
give the minimum $R_e(\tilde{\mu}_s,\tilde{\mu}_d,\epsilon)$
under the limit $\tilde{\mu}_d \to 0$.

\begin{figure}[htbp]
\centering
\includegraphics[width=8.00cm, clip]{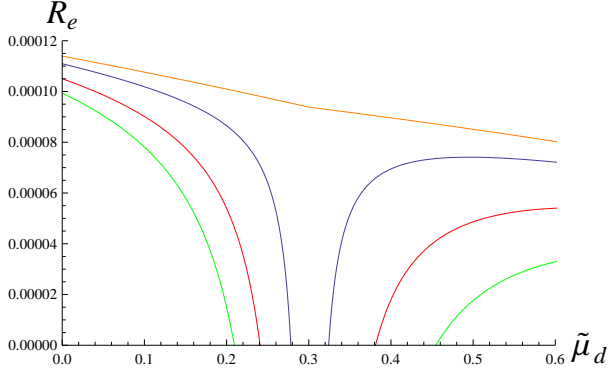}
\caption{The graphs describe the key generation rate 
${R}_e(\tilde{\mu}_s,0.3,\epsilon)$
when the signal intensity is fixed to $0.3$.
The horizontal axis expresses the decoy intensity $\tilde{\mu}_d$.
Here, the parameter $\epsilon$ describes the error rate between the true intensities and our intents,
and is chosen to be 
0\% (orange),
1\% (blue), 3\% (red), and 5\% (green).} 
\Label{f6}
\end{figure}

As is numerically demonstrated in Fig \ref{f6},
$R_e(\tilde{\mu}_s,\tilde{\mu}_d,\epsilon) $
and $\tilde{R}_e(\tilde{\mu}_s,\tilde{\mu}_d,\epsilon) $
are not necessarily monotonically decreasing with respect to $\tilde{\mu}_d$
when $\tilde{\mu}_s < \tilde{\mu}_d$.
However, 
our numerical analysis in Fig \ref{f6}, suggests that the maximums of 
$R_e(\tilde{\mu}_s,\tilde{\mu}_d,\epsilon) $
and $\tilde{R}_e(\tilde{\mu}_s,\tilde{\mu}_d,\epsilon) $
are realized by $\tilde{\mu}_d \to 0$
with fixed $\tilde{\mu}_s$ and $\epsilon>0$.
This implication can be shown as the following theorem,
which will be shown in Appendix \ref{a6}.

\begin{thm}\Label{thm20}
When a fixed intensity $\tilde{\mu}_s$
satisfies that
\begin{align}
\tilde{\mu}_s(1+\epsilon) \le
1- 
\frac{p_0}{(
\frac{\alpha}{1+\epsilon}
+p_0)
 (1-h(
\frac{s\frac{\alpha}{1+\epsilon}+\frac{p_0}{2}}
{\frac{\alpha}{1+\epsilon}+p_0}
)
},\Label{11-7-9}
\end{align}
we obtain
\begin{align}
 &
{R}_e(\tilde{\mu}_s,\epsilon)
:=
\sup_{\tilde{\mu}_d:\tilde{\mu}_d< \tilde{\mu}_s} 
{R}_e(\tilde{\mu}_s,\tilde{\mu}_d,\epsilon)
=
\lim_{\tilde{\mu}_d \to 0} 
{R}_e(\tilde{\mu}_s,\tilde{\mu}_d,\epsilon) \nonumber\\
=&
\lim_{\tilde{\mu}_d \to 0} 
R_e(\tilde{\mu}_s,\tilde{\mu}_d,(1-\epsilon)\tilde{\mu}_s,(1+\epsilon)\tilde{\mu}_d) \nonumber \\
=&
(1-\epsilon)\tilde{\mu}_s e^{-(1-\epsilon)\tilde{\mu}_s} 
(
\frac{\alpha}{1+\epsilon}
+p_0)
 (1-h(
\frac{s\frac{\alpha}{1+\epsilon}+\frac{p_0}{2}}
{\frac{\alpha}{1+\epsilon}+p_0}
)
) \nonumber \\
&+ e^{-(1-\epsilon)\tilde{\mu}_s} p_0 
- p_{s,+} \eta h (e_{s,+}),\Label{10-24-2} \\
 &
\tilde{R}_e(\tilde{\mu}_s,\epsilon)
:=
\sup_{\tilde{\mu}_d:\tilde{\mu}_d< \tilde{\mu}_s} 
\tilde{R}_e(\tilde{\mu}_s,\tilde{\mu}_d,\epsilon)
=
\lim_{\tilde{\mu}_d \to 0} 
\tilde{R}_e(\tilde{\mu}_s,\tilde{\mu}_d,\epsilon) \nonumber\\
=&
\lim_{\tilde{\mu}_d \to 0} 
R_e(\tilde{\mu}_s,\tilde{\mu}_d,(1-\epsilon)\tilde{\mu}_s,(1+\epsilon)\tilde{\mu}_d) \nonumber \\
=&
(1-\epsilon)\tilde{\mu}_s e^{-(1-\epsilon)\tilde{\mu}_s} 
(
\frac{\alpha}{1+\epsilon}
+p_0)
 (1-h(
\frac{s\frac{\alpha}{1+\epsilon}+\frac{p_0}{2}}
{\frac{\alpha}{1+\epsilon}+p_0}
)
) \nonumber \\
&+ e^{-(1-\epsilon)\tilde{\mu}_s} p_0 
- p_{s,+} \eta h (e_{s,+}). \Label{10-24-3} 
\end{align}
\end{thm}

This theorem implies that 
the infinitesimal small decoy intensity gives the best key generation rate.
Using this theorem,
we numerically demonstrate ${R}_e(\tilde{\mu}_s,\epsilon)$ in Fig \ref{f7}.
Then, we find the optimal signal intensity for our method as in Table \ref{T1}.

\begin{figure}[htbp]
\centering
\includegraphics[width=8.00cm, clip]{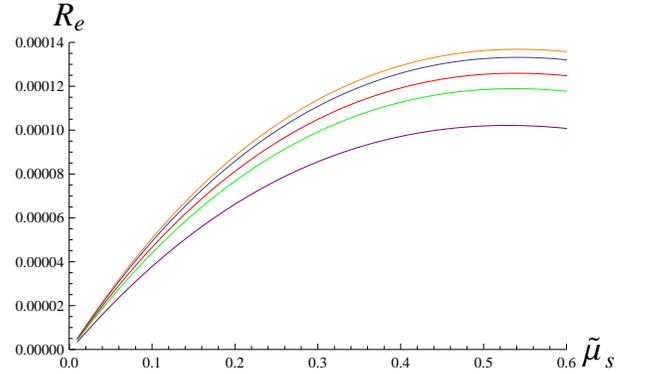}
\caption{The graphs describe the optimal key generation rate 
${R}_e(\tilde{\mu}_s,\epsilon)$.
The horizontal axis expresses the signal intensity $\tilde{\mu}_s$.
Here, the parameter $\epsilon$ describes the error rate between the true intensities and our intents,
and is chosen to be 
0\% (orange),
1\% (blue), 3\% (red), 5\% (green), and 10\% (purple).} 
\Label{f7}
\end{figure}

\begin{table}[htbp]
\centering
\caption{
The optimal key generation rate ${R}_e(\tilde{\mu}_s,+0,\epsilon)$,
and 
the optimal signal intensities $\tilde{\mu}_s$ 
when the decoy intensity $\tilde{\mu}_d$ is infinitesimal.
Here, the parameter $\epsilon$ describes the error rate between the true intensities and our intents.
}\Label{T1}
\begin{tabular}{|l|l|l|l|}
\hline
$\epsilon$ &$\tilde{\mu}_s$ & ${R}_e$ \\
\hline
0\% & 0.539023 &  0.000136994 \\
\hline
1\% & 0.539212 & 0.000133318 \\
\hline
3\% & 0.539293 & 0.000125944 \\
\hline
5\% & 0.535419 & 0.000119117 \\
\hline
10\% & 0.528461 &  0.000102458\\
\hline
\end{tabular}
\end{table}

\section{Relation with the finite-length case}\Label{s4}
However, in the realistic setting,
we have to care about the length of our code.
That is, we have to estimate 
the parameters $a$, $b$, and $c$ from the finite number of pulses.
Such a case has been discussed in the recent paper \cite{HN}.
Due to the analysis in \cite{HN},
the errors of the estimates $a$, $b$, and $c$
become large when the decoy intensity is close to zero.
So, we cannot say that a smaller decoy intensity is better in the real implementation.
However, when the size of code is sufficiently large,
we can expect that the contribution of such errors is not so large.
To verify this implication,
we numerically compare our asymptotic key generation rate 
$R(\mu_d,\mu_s)$ with the rates given in \cite{HN} as Fig. \ref{f4}.
The numerical comparison suggests that 
the finite-length case has a trend similar to the asymptotic case.
The paper \cite{AT11} reports that privacy amplification has been implemented with the bit-length of raw keys up to $2 \times 10^{6}$.
However, there is a possibility to improve the method \cite{AT11}.
The forthcoming paper \cite{HT} will propose a new algorithm to realize secure hash functions.
Combining the method \cite{AT11} and the algorithm \cite{HT}, 
the bit-length of raw keys was increased up to $10^{8}$ \cite{Tsutu}.
So, we can conclude that our asymptotic analysis has reflects the trends of realizable finite-length codes.

\begin{figure}[htbp]
\centering
\includegraphics[width=8.00cm, clip]{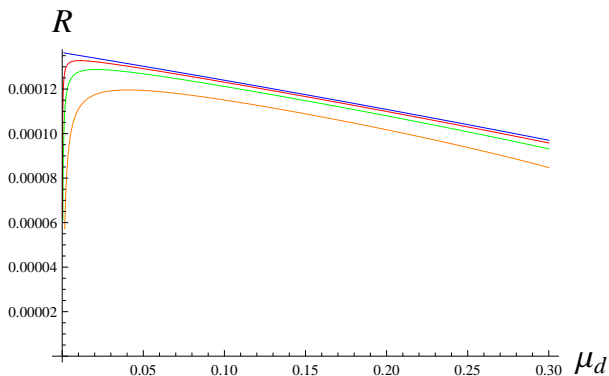}
\caption{The blue graph describes the asymptotic key generation rate 
$R(0.5,\mu_d)$ with the signal intensity $0.5$.
The horizontal axis expresses the decoy intensity $\mu_d$.
Other graphs describe the key generation rate with finite-length codes, in which,
the number of each transmitted decoy pulses is one tenth of the number of the transmitted signal pulses.
The bit-length of raw keys is chosen to be 
$10^{8}$ (orange), 
$10^{9}$ (green), and
$10^{10}$ (red).}
\Label{f4}
\end{figure}

\section{Conclusion}\Label{s5}
First, we have improved the decoy protocol 
when we employ only one decoy intensity and the vacuum pulse
by introducing the new parameterization of the channel.
Then, 
in both the existing method and our improved method,
we have shown that 
a smaller decoy intensity gives a larger 
key generation rate in the asymptotic setting.
Hence, the infinitesimal decoy intensity realizes the optimal asymptotic 
key generation rate,
and yields the perfect estimation of the counting rate and the phase error rate 
of the single photon pulse.
We also verify the latter conclusion 
even when we cannot control or identify 
the intensities ${\mu}_d$ and ${\mu}_s$
by assuming an assumption similar to Wang et al.\cite{wang2,wang3}.
Then, we have numerically optimized the signal intensity under the optimal decoy intensity.
Finally, we have numerically checked that
our conclusion is almost valid even for finite-length code \cite{HN}.

\subsection*{Acknowledgment}
The author thanks 
Prof. Akihisa Tomita, Dr. Toyohiro Tsurumaru,
and Mr. Ryota Nakayama
for valuable comments.
He is partially supported by a MEXT Grant-in-Aid for Scientific Research (A) No. 23246071. 
He is also partially supported by the National Institute of Information and Communication Technology (NICT), Japan.
The Centre for Quantum Technologies is funded by the Singapore Ministry of Education and the National Research Foundation as part of the Research Centres of Excellence programme.


\appendix

\section{Proof of Lemmas \ref{thm2} and \ref{thm22}}\Label{as1}
Define the function $f(c,b):=
(c+b)(1-h (\frac{b}{c+b}))$.
Since
the assumption implies 
\begin{align*}
\frac{\partial f}{\partial c}
= 1+ \log \frac{c}{c+b} >0,\quad
\frac{\partial f}{\partial b}
= 1+ \log \frac{b}{c+b} <0,
\end{align*}
in order to show the first argument of Lemma \ref{thm2},
it is sufficient to show that
$\hat{c}(\mu_1,\mu_2)$
is monotonically decreasing for $\mu_1$ and $\mu_2$
and 
$\hat{b}(\mu_1)$
is monotonically increasing for $\mu_1$.
Similarly, 
define the function $g(a,b):=
a(1-h (\frac{b}{a}))$.
Since
the assumption implies 
\begin{align*}
\frac{\partial g}{\partial a}
= 1+ \log \frac{a-b}{a} >0,\quad
\frac{\partial f}{\partial b}
= \log \frac{b}{a-b} <0,
\end{align*}
in order to show the second argument of Lemma \ref{thm2},
it is sufficient to show that
$\hat{a}(\mu_1,\mu_2)$
is monotonically decreasing for $\mu_1$ and $\mu_2$
and 
$\hat{b}(\mu_1)$
is monotonically increasing for $\mu_1$.

Since
\begin{align}
 & \hat{b}(\mu_1)
=
\frac{s(e^{\mu_1}-e^{(1-\alpha )\mu_1}) 
+p_0 (e^{\mu_1}-1) /2}{\mu_1} \nonumber \\
 =&
s \alpha +\frac{p_0}{2}+
\sum_{n=2}^{\infty}
(s (1-(1-\alpha )^n) +\frac{p_0}{2} )
\frac{\mu_1^{n-1}}{n!}
\Label{10-27-1}
\end{align}
and 
$s (1-(1-\alpha )^n) +\frac{p_0}{2} \ge 0$,
$\hat{b}(\mu_1)$
is monotonically increasing for $\mu_1$.

\begin{widetext}
Further,
\begin{align}
& \hat{c}(\mu_1,\mu_2)
=
(\mu_2 \mu_1)
\frac
{
\frac{(1-s+p_0/2) e^{\mu_1} -(1-s)e^{ (1-\alpha) \mu_1}-p_0/2}{\mu_1^2}
-
\frac{(1-s+p_0/2) e^{\mu_2} -(1-s)e^{ (1-\alpha) \mu_2}-p_0/2}{\mu_2^2}
}
{\mu_2-\mu_1 } \nonumber \\
=&
(\mu_2 \mu_1)
\frac
{
\sum_{n=1}^{\infty}
\frac{(1-s+p_0/2) \mu_1^{n-2}  -(1-s)(1-\alpha)^n \mu_1^{n-2} }{n ! }
-
\sum_{n=1}^{\infty}
\frac{(1-s+p_0/2) \mu_2^{n-2}  -(1-s)(1-\alpha)^n \mu_2^{n-2} }{n ! }
}
{\mu_2-\mu_1 } \nonumber \\
=&
(\mu_2 \mu_1)
\sum_{n=1}^{\infty}
\frac{(1-s+p_0/2 -(1-s)(1-\alpha)^n) (\mu_1^{n-2}-\mu_2^{n-2}) }{n ! (\mu_2-\mu_1)}\\
 =&
p_0/2+(1-s)\alpha
+
(\mu_2 \mu_1)
\sum_{n=3}^{\infty}
\frac{(1-s+p_0/2 -(1-s)(1-\alpha)^n ) (\mu_1^{n-2}-\mu_2^{n-2}) }{n ! (\mu_2-\mu_1)}\nonumber \\
 =&
p_0/2+(1-s)\alpha
-
(\mu_2 \mu_1)
\sum_{n=3}^{\infty}
\frac{1-s+p_0/2 -(1-s)(1-\alpha)^n }{n !} 
(\sum_{m=0}^{n-3}
\mu_1^m \mu_2^{n-3-m} )\nonumber \\
 =&
p_0/2+(1-s)\alpha
-
\sum_{n=3}^{\infty}
\frac{(1-s)(1-(1-\alpha)^n)+p_0/2  }{n !} 
(\sum_{m=0}^{n-3}
\mu_1^{m+1} \mu_2^{n-2-m} )\Label{10-23-1}.
\end{align}
\end{widetext}
Here,
$\frac{(1-s)(1-(1-\alpha)^n) +p_0/2 }{n !} $ is always positive.
Hence, $\hat{c}(\mu_1,\mu_2)$
is monotonically decreasing for $\mu_1$ and $\mu_2$.

Since similar to (\ref{10-23-1}), 
we have 
\begin{align}
&\hat{a}(\mu_1,\mu_2) \nonumber \\
 =&
p_0+\alpha \nonumber \\
&-
\sum_{n=3}^{\infty}
\frac{(1-(1-\alpha)^n)+p_0  }{n !} 
(\sum_{m=0}^{n-3}
\mu_1^{m+1} \mu_2^{n-2-m} )\Label{10-23-2},
\end{align}
$\hat{a}(\mu_1,\mu_2)$ is monotonically decreasing for $\mu_1$ and $\mu_2$.
Hence, we obtain Lemma \ref{thm2}.

Lemma \ref{thm22} can be proven by substituting $0$ into $\mu_1$
in (\ref{10-27-1}), (\ref{10-23-1}), and (\ref{10-23-2}).
\endproof

\section{Proof of Theorem \ref{thm20}}\Label{a6}
For a proof of Theorem \ref{thm20}, 
we prepare the following two lemmas
under the assumption that $\mu_d < \mu_s$.

\begin{widetext}
\begin{lem}\Label{lem22}
Assume that
\begin{align}
\hat{a}&=f_c(1-e^{-\alpha \tilde{\mu}_d} + p_0,1-e^{-\alpha \tilde{\mu}_s} + p_0,\mu_d,\mu_s) 
=
\frac{\mu_d^2 e^{\mu_d} (1-e^{-\alpha \tilde{\mu}_d}+ p_0)
-\mu_s^2 e^{\mu_s}(1-e^{-\alpha \tilde{\mu}_s}+ p_0)}
{\mu_d \mu_s (\mu_s-\mu_d) }, \\
\hat{b}&=f_b(s(1-e^{-\alpha \tilde{\mu}_d}) + \frac{p_0}{2},\mu_d)
=
\frac{
(s(1-e^{-\alpha \tilde{\mu}_d}) + \frac{p_0}{2})e^{\mu_d}-p_0/2
}{\mu_d},
\\
\hat{c}&=f_c((1-s)(1-e^{-\alpha \tilde{\mu}_d}) + \frac{p_0}{2},(1-s)(1-e^{-\alpha \tilde{\mu}_s}) + \frac{p_0}{2},\mu_d,\mu_s) \nonumber \\
&=
\frac{\mu_d^2 e^{\mu_d} ( (1-s)(1-e^{-\alpha \tilde{\mu}_d})+ p_0/2)
-\mu_s^2 e^{\mu_2}( (1-s)1-e^{-\alpha \tilde{\mu}_s}+ p_0/2)}
{\mu_d \mu_s (\mu_s-\mu_d) }.
\end{align}
Then,
\begin{enumerate}[(i)]
\item 
$\hat{a}$ is monotonically increasing with respect to $\tilde{\mu}_d$
and 
monotonically decreasing with respect to $\tilde{\mu}_s$.

\item 
$\hat{b}$ is monotonically increasing with respect to $\tilde{\mu}_d$.

\item 
$\hat{c}$ is monotonically increasing with respect to $\tilde{\mu}_d$
and 
monotonically decreasing with respect to $\tilde{\mu}_s$.

\item 
$\frac{\hat{b}}{\hat{a}}$ is 
monotonically decreasing with respect to $\tilde{\mu}_d$.

\item 
$\frac{\hat{b}}{\hat{c}+\hat{b}}$ is 
monotonically decreasing with respect to $\tilde{\mu}_d$.

\item 
$\hat{a}
(1-h(
\frac{\hat{b}}{\hat{a}}))$ 
is monotonically increasing with respect to $\tilde{\mu}_d$
and monotonically decreasing with respect to $\tilde{\mu}_s$.

\item 
$(\hat{c}+\hat{b})
(1-h(
\frac{\hat{b}}{\hat{c} +\hat{b}}))$ 
is monotonically increasing with respect to $\tilde{\mu}_d$
and monotonically decreasing with respect to $\tilde{\mu}_s$.
\end{enumerate}
\end{lem}

\begin{proof}
First, we notice that
$(1-e^{-\alpha \tilde{\mu}_i})$ is monotonically increasing with respect to $\tilde{\mu}_i$ for $i=s,d$.
Using this fact, we can show the items (i), (ii), and (iii).
 
Next, we will show (iv). 
Since
\begin{align*}
 \hat{b}
=
\frac{\mu_s (\mu_s-\mu_d) e^{\mu_d} s(1-e^{-\alpha \tilde{\mu}_d})+\mu_s (\mu_s-\mu_d) e^{\mu_d}(1-e^{-\mu_d})p_0/2}
{\mu_d \mu_s (\mu_s-\mu_d) },
\end{align*}
we have
\begin{align*}
 \frac{\hat{b} }{\hat{a}} 
&=
\frac{\mu_s (\mu_s-\mu_d) e^{\mu_d} s(1-e^{-\alpha \tilde{\mu}_d})+\mu_s (\mu_s-\mu_d) e^{\mu_d}(1-e^{-\mu_d})p_0/2}{\mu_d^2 e^{\mu_d} (1-e^{-\alpha \tilde{\mu}_d}+ p_0)
-\mu_s^2 e^{\mu_s}(1-e^{-\alpha \tilde{\mu}_s}+ p_0)} \\
&=
\frac{\mu_s (\mu_s-\mu_d) e^{\mu_d} s(1-e^{-\alpha \tilde{\mu}_d})+\mu_s (\mu_s-\mu_d) e^{\mu_d}(1-e^{-\mu_d})p_0/2}
{\mu_d^2 e^{\mu_d} (1-e^{-\alpha \tilde{\mu}_d})
-\mu_s^2 e^{\mu_s}(1-e^{-\alpha \tilde{\mu}_s})
+p_0 (\mu_d^2 e^{\mu_d}-\mu_s^2 e^{\mu_s})
} \\
&=
\frac{A_1(1-e^{-\alpha \tilde{\mu}_d})+A_2}
{ A_3(1-e^{-\alpha \tilde{\mu}_d})
-A_4(1-e^{-\alpha \tilde{\mu}_s})+A_5},
\end{align*}
where
\begin{align*}
A_1&:=\mu_s (\mu_s-\mu_d) e^{\mu_d} s ,\quad
A_2:=\mu_s (\mu_s-\mu_d) e^{\mu_d}(1-e^{-\mu_d})p_0/2, \\
A_3&:=\mu_d^2 e^{\mu_d}, \quad
A_4:=\mu_s^2 e^{\mu_s}, \quad
A_5:=p_0 (\mu_d^2 e^{\mu_d}-\mu_s^2 e^{\mu_s}).
\end{align*}
As is shown below, we have
$ \frac{A_1}{A_3}
\le \frac{A_2}{A_5-A_4(1-e^{-\alpha \tilde{\mu}_s})}$.
Thus, 
since $1-e^{-\alpha \tilde{\mu}_d}$ is monotonically increasing with respect to $\tilde{\mu}_d$,
$\frac{A_1(1-e^{-\alpha \tilde{\mu}_d})+A_2}
{ A_3(1-e^{-\alpha \tilde{\mu}_d})
-A_4(1-e^{-\alpha \tilde{\mu}_s})
+A_5}$
is monotonically decreasing with respect to $\tilde{\mu}_d$.
Then, we obtain the item (iv).
Now, we will show 
$ \frac{A_1}{A_3}
\le \frac{A_2}{A_5-A_4(1-e^{-\alpha \tilde{\mu}_s})}$.
Since $s \le 1/2$
and
$\mu_s^2 e^{\mu_s}\ge \mu_d^2$,
\begin{align*}
& \mu_d^2 e^{\mu_d} {(1-e^{-\mu_d})p_0/2}
\ge
\mu_d^2 e^{\mu_d} {s (1-e^{-\mu_d})p_0}
= s p_0 (\mu_d^2 e^{\mu_d} -s \mu_d^2 ) \\
\ge & s p_0 (\mu_d^2 e^{\mu_d}-\mu_s^2 e^{\mu_s})
\ge 
s p_0 (\mu_d^2 e^{\mu_d}-\mu_s^2 e^{\mu_s})
-s\mu_s^2 e^{\mu_s}(1-e^{-\alpha \tilde{\mu}_s}).
\end{align*}
Hence,
\begin{align*}
\frac{s}{\mu_d^2 e^{\mu_d} } 
\le 
\frac
{(1-e^{-\mu_d})p_0/2}
{p_0 (\mu_d^2 e^{\mu_d}-\mu_s^2 e^{\mu_s})
-\mu_s^2 e^{\mu_s}(1-e^{-\alpha \tilde{\mu}_s})}.
\end{align*}
Multiplying $\mu_s (\mu_s-\mu_d)$,
we have
\begin{align*}
& \frac{A_1}{A_3}
=\frac{\mu_s (\mu_s-\mu_d) e^{\mu_d} s}{\mu_d^2 e^{\mu_d} } 
\le 
\frac{\mu_s (\mu_s-\mu_d) e^{\mu_d}(1-e^{-\mu_d})p_0/2}
{p_0 (\mu_d^2 e^{\mu_d}-\mu_s^2 e^{\mu_s})
-\mu_s^2 e^{\mu_s}(1-e^{-\alpha \tilde{\mu}_s})}
=
\frac{A_2}{A_5-A_4(1-e^{-\alpha \tilde{\mu}_s})}.
\end{align*}

Next, we will show (iv). 
We have
\begin{align*}
 \frac{\hat{b} }{\hat{c}+\hat{b} } 
&=
\frac{\mu_s e^{\mu_d} (\mu_s-\mu_d) s(1-e^{-\alpha \tilde{\mu}_d})+\mu_s (\mu_s-\mu_d) e^{\mu_d}(1-e^{-\mu_d})p_0/2}
{\mu_s e^{\mu_d}(\mu_s-\mu_d s)(1-e^{-\alpha \tilde{\mu}_d})+
\mu_s e^{\mu_d}(1- e^{-\mu_d})(\mu_s-\mu_d/2)p_0 -B_1} \\
 &=
\frac{B_2(1-e^{-\alpha \tilde{\mu}_d})+B_3}{B_4(1-e^{-\alpha \tilde{\mu}_d})+B_5},
\end{align*}
where
\begin{align*}
B_1:=&
\mu_d^2 e^{\mu_s}
((1-s)(1-e^{-\alpha \tilde{\mu}_s})+ (1- e^{-\mu_s})p_0 /2) \\
B_2:=& \mu_s e^{\mu_d} (\mu_s-\mu_d) s \\
B_3:=& \mu_s (\mu_s-\mu_d) e^{\mu_d}(1-e^{-\mu_d})p_0/2\\ 
B_4:=&\mu_s e^{\mu_d}(\mu_s-\mu_d s) \\
B_5:=& \mu_s e^{\mu_d}(1- e^{-\mu_d})(\mu_s-\mu_d/2)p_0 -B_1.
\end{align*}
As is shown below, we have
$ \frac{B_2}{B_4} \le \frac{B_3}{B_5}$.
Thus, since $1-e^{-\alpha \tilde{\mu}_d}$ is monotonically increasing with respect to $\tilde{\mu}_d$,
$\frac{B_2(1-e^{-\alpha \tilde{\mu}_d})+B_3}{B_4(1-e^{-\alpha \tilde{\mu}_d})+B_5}$
is monotonically decreasing with respect to $\tilde{\mu}_d$.
Then, we obtain the desired argument for $\frac{\hat{b} }{\hat{c}+\hat{b}}$.

Now, we will show $ \frac{B_2}{B_4} \le \frac{B_3}{B_5}$.
Since $s< \frac{1}{2}$,
we have $
\frac{s}{(\mu_s-\mu_d s)} 
< 
\frac{1}{2(\mu_s-\mu_d/2)}$,
which implies
\begin{align*}
\frac{B_2}{B_4}&=
\frac{\mu_s e^{\mu_d} (\mu_s-\mu_d) s}{\mu_s e^{\mu_d}(\mu_s-\mu_d s)}
=
\frac{s(\mu_s-\mu_d)}{(\mu_s-\mu_d s)} 
< 
\frac{\mu_s-\mu_d}{2(\mu_s-\mu_d/2)} \\
 &=\frac{\mu_s (\mu_s-\mu_d) e^{\mu_d}(1-e^{-\mu_d})p_0/2}{\mu_s e^{\mu_d}(1- e^{-\mu_d})(\mu_s-\mu_d/2) p_0} 
<
\frac{\mu_s (\mu_s-\mu_d) e^{\mu_d}(1-e^{-\mu_d})p_0/2}{\mu_s e^{\mu_d}(1- e^{-\mu_d})(\mu_s-\mu_d/2)p_0 -B_1}
=\frac{B_3}{B_5}.
\end{align*}

For a proof of (vi), we use the function $g(a,b)$ defined in the proof of Theorem \ref{thm2}.
Since $g(a,b)$ is monotonically increasing with respect to $a$ and 
monotonically decreasing with respect to $b$,
we obtain the item (vi).
The item (vii) can be shown in the same way by replacing 
the function $g(a,b)$ by the function $f(c,b)$ defined in the proof of Theorem \ref{thm2}.
\end{proof}

\begin{lem}\Label{lem23}
When $\mu_s\le (1+\epsilon)\tilde{\mu}_s  $,
\begin{align}
 & \sup_{\mu_d: 0 < \mu_d< \mu_s}
f_a((1-e^{-\alpha (1+\epsilon)^{-1}\mu_d}) + p_0,
(1-e^{-\alpha \tilde{\mu}_s}) + p_0,\mu_d,\mu_s) 
\nonumber \\
 =&
\lim_{\mu_d \to +0}
f_a((1-e^{-\alpha (1+\epsilon)^{-1}\mu_d}) + p_0,
(1-e^{-\alpha \tilde{\mu}_s}) + p_0,\mu_d,\mu_s)
 =
\alpha (1+\epsilon)^{-1} +p_0,
\Label{11-7-3}\\
 & \sup_{0 < \mu_d}
f_b(s(1-e^{-\alpha (1+\epsilon)^{-1}\mu_d}) + \frac{p_0}{2},\mu_d) \nonumber \\
 =&\lim_{\mu_d \to +0}
f_b(s(1-e^{-\alpha (1+\epsilon)^{-1}\mu_d}) + \frac{p_0}{2},\mu_d) 
=
s\alpha (1+\epsilon)^{-1} +\frac{p_0}{2},\Label{11-7-2}\\
 & \sup_{\mu_d: 0 < \mu_d< \mu_s}
f_c((1-s)(1-e^{-\alpha (1+\epsilon)^{-1}\mu_d}) + \frac{p_0}{2},(1-s)(1-e^{-\alpha \tilde{\mu}_s}) + \frac{p_0}{2},\mu_d,\mu_s) \nonumber \\
 =&
\lim_{\mu_d \to +0}
f_c((1-s)(1-e^{-\alpha (1+\epsilon)^{-1}\mu_d}) + \frac{p_0}{2},(1-s)(1-e^{-\alpha \tilde{\mu}_s}) + \frac{p_0}{2},\mu_d,\mu_s)
 =
(1-s)\alpha (1+\epsilon)^{-1} +\frac{p_0}{2}.
\Label{11-7-1}
\end{align}
\end{lem}
\begin{proof}
Since 
$f_b(s(1-e^{-\alpha (1+\epsilon)^{-1}\mu_d}) + \frac{p_0}{2},\mu_d) 
=
\frac{
(s(1-e^{-\alpha (1+\epsilon)^{-1}\mu_d}) + \frac{p_0}{2})e^{\mu_d}-p_0/2
}{\mu_d}$ 
is monotonically decreasing with respect to $\mu_d$,
we obtain (\ref{11-7-2}).
Similarly, it is sufficient for (\ref{11-7-1}) to show that
$f_c((1-s)(1-e^{-\alpha (1+\epsilon)^{-1}\mu_d}) + \frac{p_0}{2},(1-s)(1-e^{-\alpha \tilde{\mu}_s}) + \frac{p_0}{2},\mu_d,\mu_s)$ 
is monotonically decreasing with respect to $\mu_d$.
Choosing $\tilde{\alpha}:=\alpha (1+\epsilon)^{-1}$,
we obtain
\begin{align*}
 &f_c((1-s)(1-e^{-\alpha (1+\epsilon)^{-1}\mu_d}) + \frac{p_0}{2},(1-s)(1-e^{-\alpha \tilde{\mu}_s}) + \frac{p_0}{2},\mu_d,\mu_s) \\
 =&
\frac
{
\mu_s^2 e^{\mu_d}
((1-s)(1-e^{-\tilde{\alpha} \mu_d})+ p_0 (1- e^{-\mu_d})/2 )
-
\mu_d^2 e^{\mu_s}
((1-s)(1-e^{-\alpha \mu_s})+ p_0 (1- e^{-\mu_s})/2)
}
{\mu_d \mu_s (\mu_s-\mu_d) } \\
 =&
\frac
{
\mu_s^2 e^{\mu_d}
((1-s)(1-e^{-\tilde{\alpha} \mu_d})+ p_0 (1- e^{-\mu_d})/2 )
-
\mu_d^2 e^{\mu_s}
((1-s)(1-e^{-\tilde{\alpha} \mu_s})+ p_0 (1- e^{-\mu_s})/2)
}
{\mu_d \mu_s (\mu_s-\mu_d) } \\
 &
-\frac
{\mu_d e^{\mu_s}
((1-s)(e^{-\tilde{\alpha} \mu_s}-e^{-\alpha \tilde{\mu}_s})}
{\mu_s (\mu_s-\mu_d) } \\
 =&
p_0/2+(1-s)\tilde{\alpha}
-
\sum_{n=3}^{\infty}
\frac{(1-s)(1-(1-\alpha)^n)+p_0/2  }{n !} 
(\sum_{m=0}^{n-3}
\mu_d^{m+1} \mu_s^{n-2-m} )
-\frac
{\mu_d e^{\mu_s}
(1-s)(e^{-\tilde{\alpha} \mu_s}-e^{-\alpha \tilde{\mu}_s})}
{\mu_s (\mu_s-\mu_d) } .
\end{align*}
Since $e^{-\tilde{\alpha} \mu_s}-e^{-\alpha \tilde{\mu}_s}\ge 0$,
the final term is monotonically decreasing with respect to $\mu_d$.
Other terms are also monotonically decreasing with respect to $\mu_d$.

Similarly, we have
\begin{align*}
 &
f_a((1-e^{-\alpha (1+\epsilon)^{-1}\mu_d}) + p_0,
(1-e^{-\alpha \tilde{\mu}_s}) + p_0,\mu_d,\mu_s) \\
 =&
p_0+\tilde{\alpha}
-
\sum_{n=3}^{\infty}
\frac{(1-(1-\alpha)^n)+p_0  }{n !} 
(\sum_{m=0}^{n-3}
\mu_d^{m+1} \mu_s^{n-2-m} )
-\frac
{\mu_d e^{\mu_s}
(e^{-\tilde{\alpha} \mu_s}-e^{-\alpha \tilde{\mu}_s})}
{\mu_s (\mu_s-\mu_d) } .
\end{align*}
So,
$f_a((1-e^{-\alpha (1+\epsilon)^{-1}\mu_d}) + p_0,
(1-e^{-\alpha \tilde{\mu}_s}) + p_0,\mu_d,\mu_s)$
is also monotonically decreasing with respect to $\mu_d$.
Therefore, we obtain (\ref{11-7-3}).
\end{proof}

\begin{proofof}{Theorem \ref{thm20}}
First, we show the case of ${R}_e(\tilde{\mu}_s,\tilde{\mu}_d,\epsilon) $.
For a fixed $\tilde{\mu_s}$ and $\mu_s$ satisfying 
$(1-\epsilon)\tilde{\mu}_s< {\mu}_s< (1+\epsilon)\tilde{\mu}_s$,
Lemmas \ref{lem22} and \ref{lem23} imply 
\begin{align}
 & 
\sup_{\tilde{\mu_d}: 0<\tilde{\mu_d}< \frac{1-\epsilon}{1+\epsilon}\tilde{\mu_s}}
\min_{\mu_d \in [(1-\epsilon)\tilde{\mu}_d,(1+\epsilon)\tilde{\mu}_d]}
R_e({\mu}_s,\mu_d,\tilde{\mu}_s,\tilde{\mu}_d) 
=
\sup_{\mu_d: 0 < \mu_d< \mu_s}
\min_{
\tilde{\mu}_d:
(1-\epsilon)\tilde{\mu}_d< {\mu}_d< (1+\epsilon)\tilde{\mu}_d}
R_e({\mu}_s,\mu_d,\tilde{\mu}_s,\tilde{\mu}_d) \nonumber \\
 &=
\sup_{\mu_d: 0 < \mu_d< \mu_s}
R_e({\mu}_s,\mu_d,\tilde{\mu}_s,(1+\epsilon)^{-1}{\mu}_d) 
=
\lim_{\mu_d \to +0}
R_e(\mu_s,{\mu}_d,\tilde{\mu}_s,(1+\epsilon)^{-1}{\mu}_d).
 \Label{11-7-4}
\end{align}
Hence, we obtain
\begin{align}
 & 
\sup_{\tilde{\mu_d}: 0<\tilde{\mu_d}< \tilde{\mu_s}}
\min_{\mu_d \in [(1-\epsilon)\tilde{\mu}_d,(1+\epsilon)\tilde{\mu}_d]}
R_e(\mu_s,{\mu}_d,\tilde{\mu}_s,\tilde{\mu}_d) 
 =
\lim_{\tilde{\mu_d} \to +0}
\min_{\mu_d \in [(1-\epsilon)\tilde{\mu}_d,(1+\epsilon)\tilde{\mu}_d]}
R_e(\mu_s,{\mu}_d,\tilde{\mu}_s,\tilde{\mu}_d) \Label{11-7-5} \\
 =&
\lim_{\tilde{\mu}_d \to 0} 
R_e(\mu_s,(1+\epsilon)\tilde{\mu}_d,\tilde{\mu}_s,\tilde{\mu}_d) .\nonumber 
\end{align}
Since the convergence (\ref{11-7-4}) is uniform with respect to $\mu_s$
and $\tilde{\mu}_s$,
the convergence (\ref{11-7-5}) is uniform with respect to $\mu_s$
and $\tilde{\mu}_s$.
Hence, we obtain
\begin{align}
 &
\min_{\mu_s \in [(1-\epsilon)\tilde{\mu}_s,(1+\epsilon)\tilde{\mu}_s]}
\lim_{\tilde{\mu_d} \to +0}
\min_{\mu_d \in [(1-\epsilon)\tilde{\mu}_d,(1+\epsilon)\tilde{\mu}_d]}
R_e(\mu_s,{\mu}_d,\tilde{\mu}_s,\tilde{\mu}_d) \nonumber \\
 =&
\lim_{\tilde{\mu_d} \to +0}
\min_{\mu_s \in [(1-\epsilon)\tilde{\mu}_s,(1+\epsilon)\tilde{\mu}_s]}
\min_{\mu_d \in [(1-\epsilon)\tilde{\mu}_d,(1+\epsilon)\tilde{\mu}_d]}
R_e(\mu_s,{\mu}_d,\tilde{\mu}_s,\tilde{\mu}_d) .
\Label{11-7-6}
\end{align}

On the other hand,
\begin{align}
 &
\min_{\mu_s \in [(1-\epsilon)\tilde{\mu}_s,(1+\epsilon)\tilde{\mu}_s]}
\lim_{\tilde{\mu_d} \to +0}
\min_{\mu_d \in [(1-\epsilon)\tilde{\mu}_d,(1+\epsilon)\tilde{\mu}_d]}
R_e(\mu_s,{\mu}_d,\tilde{\mu}_s,\tilde{\mu}_d) \nonumber \\
 =&
\min_{\mu_s \in [(1-\epsilon)\tilde{\mu}_s,(1+\epsilon)\tilde{\mu}_s]}
\sup_{\tilde{\mu_d}:
0<\tilde{\mu_d}< \frac{1-\epsilon}{1+\epsilon}\tilde{\mu_s}}
\min_{\mu_d \in [(1-\epsilon)\tilde{\mu}_d,(1+\epsilon)\tilde{\mu}_d]}
R_e(\mu_s,{\mu}_d,\tilde{\mu}_s,\tilde{\mu}_d) \nonumber \\
 \ge &
\sup_{\tilde{\mu_d}:
0<\tilde{\mu_d}< \frac{1-\epsilon}{1+\epsilon}\tilde{\mu_s}}
\min_{\mu_s \in [(1-\epsilon)\tilde{\mu}_s,(1+\epsilon)\tilde{\mu}_s]}
\min_{\mu_d \in [(1-\epsilon)\tilde{\mu}_d,(1+\epsilon)\tilde{\mu}_d]}
R_e(\mu_s,{\mu}_d,\tilde{\mu}_s,\tilde{\mu}_d) \nonumber \\
 \ge &
\lim_{\tilde{\mu_d} \to +0}
\min_{\mu_s \in [(1-\epsilon)\tilde{\mu}_s,(1+\epsilon)\tilde{\mu}_s]}
\min_{\mu_d \in [(1-\epsilon)\tilde{\mu}_d,(1+\epsilon)\tilde{\mu}_d]}
R_e(\mu_s,{\mu}_d,\tilde{\mu}_s,\tilde{\mu}_d) .\Label{11-7-7}
\end{align}
Combining (\ref{11-7-6}) and (\ref{11-7-7}),
we obtain
\begin{align}
 &
\sup_{\tilde{\mu_d}:
0<\tilde{\mu_d}< \frac{1-\epsilon}{1+\epsilon}\tilde{\mu_s}}
\min_{\mu_s \in [(1-\epsilon)\tilde{\mu}_s,(1+\epsilon)\tilde{\mu}_s]}
\min_{\mu_d \in [(1-\epsilon)\tilde{\mu}_d,(1+\epsilon)\tilde{\mu}_d]}
R_e(\mu_s,{\mu}_d,\tilde{\mu}_s,\tilde{\mu}_d) \nonumber \\
 = &
\min_{\mu_s \in [(1-\epsilon)\tilde{\mu}_s,(1+\epsilon)\tilde{\mu}_s]}
\lim_{\tilde{\mu_d} \to +0}
\min_{\mu_d \in [(1-\epsilon)\tilde{\mu}_d,(1+\epsilon)\tilde{\mu}_d]}
R_e(\mu_s,{\mu}_d,\tilde{\mu}_s,\tilde{\mu}_d) \nonumber \\
 = &
\min_{\mu_s \in [(1-\epsilon)\tilde{\mu}_s,(1+\epsilon)\tilde{\mu}_s]}
\lim_{\tilde{\mu_d} \to +0}
R_e(\mu_s,(1+\epsilon)\tilde{\mu}_d,\tilde{\mu}_s,\tilde{\mu}_d) \nonumber  \\
 = &
\min_{\mu_s \in [(1-\epsilon)\tilde{\mu}_s,(1+\epsilon)\tilde{\mu}_s]}
\lim_{\mu_d \to +0}
R_e(\mu_s,{\mu}_d,\tilde{\mu}_s,(1+\epsilon)^{-1}{\mu}_d) \nonumber  \\
 = &
\min_{\mu_s \in [(1-\epsilon)\tilde{\mu}_s,(1+\epsilon)\tilde{\mu}_s]}
{\mu}_s e^{-{\mu}_s} 
(
\frac{\alpha}{1+\epsilon}
+p_0)
 (1-h(
\frac{s\frac{\alpha}{1+\epsilon}+\frac{p_0}{2}}
{\frac{\alpha}{1+\epsilon}+p_0}
)
) 
+ e^{-{\mu}_s} p_0 
- p_{s,+} \eta h (e_{s,+})\nonumber \\
 =&
(1-\epsilon)\tilde{\mu}_s e^{-(1-\epsilon)\tilde{\mu}_s} 
(
\frac{\alpha}{1+\epsilon}
+p_0)
 (1-h(
\frac{s\frac{\alpha}{1+\epsilon}+\frac{p_0}{2}}
{\frac{\alpha}{1+\epsilon}+p_0}
)
) 
+ e^{-(1-\epsilon)\tilde{\mu}_s} p_0 
- p_{s,+} \eta h (e_{s,+}),\Label{10-24-1}
\end{align}
where the final equation follows from (\ref{11-7-9}) and the following fact.
Choosing the constant $C=\frac{p_0}{
(\frac{\alpha}{1+\epsilon}
+p_0)
 (1-h(
\frac{s\frac{\alpha}{1+\epsilon}+\frac{p_0}{2}}
{\frac{\alpha}{1+\epsilon}+p_0}
))}$,
we have
\begin{align*}
{\mu}_s e^{-{\mu}_s} 
(
\frac{\alpha}{1+\epsilon}
+p_0)
 (1-h(
\frac{s\frac{\alpha}{1+\epsilon}+\frac{p_0}{2}}
{\frac{\alpha}{1+\epsilon}+p_0}
)
) 
+ e^{-{\mu}_s} p_0 =
(
\frac{\alpha}{1+\epsilon}
+p_0)
 (1-h(
\frac{s\frac{\alpha}{1+\epsilon}+\frac{p_0}{2}}
{\frac{\alpha}{1+\epsilon}+p_0}
)
) 
({\mu}_s e^{-{\mu}_s}+ C e^{-{\mu}_s}).
\end{align*}
Since the function ${\mu}_s \mapsto {\mu}_s e^{-{\mu}_s}+ C e^{-{\mu}_s}$
is monotonically increasing due to the assumption (\ref{11-7-9}),
we obtain the above equation (\ref{10-24-1}).
Thus, we obtain (\ref{10-24-2}).

Next, we show the case of $\tilde{R}_e(\tilde{\mu}_s,\tilde{\mu}_d,\epsilon) $.
Similar to (\ref{11-7-4}), using Lemmas \ref{lem22} and \ref{lem23}, 
we have
\begin{align}
 & 
\sup_{\tilde{\mu_d}: 0<\tilde{\mu_d}< \frac{1-\epsilon}{1+\epsilon}\tilde{\mu_s}}
\min_{\mu_d \in [(1-\epsilon)\tilde{\mu}_d,(1+\epsilon)\tilde{\mu}_d]}
\tilde{R}_e({\mu}_s,\mu_d,\tilde{\mu}_s,\tilde{\mu}_d) \nonumber \\
 &=
\lim_{\mu_d \to +0}
\tilde{R}_e(\mu_s,{\mu}_d,\tilde{\mu}_s,(1+\epsilon)^{-1}{\mu}_d).
 \Label{11-7-4b}.
\end{align}
We also have the same relations for
$\tilde{R}_e(\tilde{\mu}_s,\tilde{\mu}_d,\epsilon) $.
as (\ref{11-7-6}) and (\ref{11-7-7}).
Hence, we obtain 
\begin{align}
 &
\sup_{\tilde{\mu_d}:
0<\tilde{\mu_d}< \frac{1-\epsilon}{1+\epsilon}\tilde{\mu_s}}
\min_{\mu_s \in [(1-\epsilon)\tilde{\mu}_s,(1+\epsilon)\tilde{\mu}_s]}
\min_{\mu_d \in [(1-\epsilon)\tilde{\mu}_d,(1+\epsilon)\tilde{\mu}_d]}
\tilde{R}_e(\mu_s,{\mu}_d,\tilde{\mu}_s,\tilde{\mu}_d) \nonumber \\
= &
\min_{\mu_s \in [(1-\epsilon)\tilde{\mu}_s,(1+\epsilon)\tilde{\mu}_s]}
{\mu}_s e^{-{\mu}_s} 
(
\frac{\alpha}{1+\epsilon}
+p_0)
 (1-h(
\frac{s\frac{\alpha}{1+\epsilon}+\frac{p_0}{2}}
{\frac{\alpha}{1+\epsilon}+p_0}
)
) 
+ e^{-{\mu}_s} p_0 
- p_{s,+} \eta h (e_{s,+})\nonumber \\
 =&
(1-\epsilon)\tilde{\mu}_s e^{-(1-\epsilon)\tilde{\mu}_s} 
(
\frac{\alpha}{1+\epsilon}
+p_0)
 (1-h(
\frac{s\frac{\alpha}{1+\epsilon}+\frac{p_0}{2}}
{\frac{\alpha}{1+\epsilon}+p_0}
)
) 
+ e^{-(1-\epsilon)\tilde{\mu}_s} p_0 
- p_{s,+} \eta h (e_{s,+}).\nonumber 
\end{align}
Thus, we obtain (\ref{10-24-3}).
\end{proofof}
\end{widetext}

\end{document}